\documentclass[preprint,11pt,authoryear]{elsarticle}

\usepackage{xcolor}

\usepackage{amssymb}
\usepackage{amsthm}

\usepackage{amsmath}

\usepackage{setspace}

\theoremstyle{plain}
\newtheorem{thm}{Theorem}
\newtheorem{lemma}[thm]{Lemma}
\newtheorem{proposition}[thm]{Proposition}
\newtheorem{example}{Example}

\theoremstyle{plain}
\newtheorem{defn}{Definition}

\theoremstyle{remark}

\newtheorem*{rem}{Remark}

\newcommand{\R}{\ensuremath{\mathbb{R}}}
\newcommand{\E}{\ensuremath{\mathbb{E}}}
\def\DCdrpm{\textbf{DC}^{\hspace{0.2ex}(d,r_\pm)}(v)}

\def\DCdr{\textbf{DC}^{\hspace{0.2ex}(d,r)}(v)}
\def\DCr{\textbf{DC}^{\hspace{0.4ex}r}(v)}
\def\UDCdr{\textbf{UDC}^{\hspace{0.2ex}(d,r)}(v)}
\def\aS{a_{v(S)}}
\def\aS{a_{v(S)}}
\def\aN{a_{v(N)}}
\def\bS{b_{v(S)}}
\def\bN{b_{v(N)}}
\def\core{\mathcal{C}}
\def\dev{\hat{\sigma}}
\def\var{\textbf{Var}}

\journal{Games and Economic Behavior}

\begin{document}

\begin{frontmatter}



\title{Stochastic cooperative games of risk averse players and application to multiple newsvendors problem}


\author[inst1]{David Ryz\'{a}k}

\affiliation[inst1]{organization={Department of Applied Mathematics, Charles University},
            city={Prague},
            country={Czech Republic}}

\author[inst1]{Martin \v{C}ern\'{y}}

\begin{abstract}
This paper studies the stochastic setting in cooperative games and suggests a solution concept based on second order stochastic dominance (SSD), which is often applied to robustly model risk averse behaviour of players in different economic and game theoretic models as it enables to model not specified levels of risk aversion among players. The main result of the paper connects this solution concept, \emph{SSD-core}, in case of uniform distribution of the game to cores of two deterministic cooperative games. Interestingly, balancedness of both of these games and convexity of one of these implies non-emptiness of the SSD-core. The opposite implication does not, in general, hold and leads to questions about intersections of cores of two games and their relations. Finally, we present an application of the SSD-core to the multiple newsvendors problem, where we provide a characterization of risk averse behaviour of players with an interpretation in terms of the model.
\end{abstract}


\end{frontmatter}

\section{Introduction}
Risk-averse behavior is a fundamental aspect of decision-making in economics and game theory. It reflects the tendency of individuals to prefer outcomes with less uncertainty, even if they may have lower expected returns. This behavior is especially significant in cooperative games, where groups of players collaborate to achieve mutual benefits under uncertainty. Numerous models incorporate risk aversion into economic and game-theoretic analyses. For instance, the expected utility theory employs specific utility functions to model risk preferences, while risk measures like variance or Value-at-Risk are commonly used in finance and operations research to account for risk-averse decision, (\cite{arrow1965aspects, pratt1964risk}).

Second-order stochastic dominance (SSD) offers a robust and general approach to modeling risk-averse behavior without specifying exact levels of risk aversion. SSD provides a partial ordering of random variables that aligns with the preferences of all risk-averse decision-makers (\cite{hadar1969rules, rothschild1970increasing}). It has been extensively used in portfolio selection, insurance, and decision analysis to compare uncertain prospects when precise utility functions are unknown or difficult to specify (\cite{levy1992stochastic, post2003empirical}).

One of the pioneering works of cooperative games with stochastic characteristic functions is by
\cite{charnes1973prior}, who introduced solution concepts based on chance-constrained programming. Their approach considers randomness in the characteristic function but relies on specific probabilistic constraints. In the context of incorporating risk aversion, \cite{suijs1999cooperative} made a significant contribution by introducing core-like solution concepts for cooperative games with stochastic characteristic functions, grounded in specific risk preferences. They reduced payoffs to functions of real numbers and random variables, using preferences that induce a total order on random variables. For example, they considered preferences of the form 
$\E [X]+b\cdot \text{Var}(X)$, where $b \in \R$ represents a risk parameter, capturing both risk-averse ($b < 0$) and risk-loving ($b>0$) behaviors. However, determining the appropriate $b$ for each player can be impractical or resource-intensive, and the analysis depends on the chosen $b$. \cite{fernandez2002cores} explored cores under stochastic order, which is related to stochastic dominance. Their work focused on different questions but highlighted the importance of ordering random variables in cooperative games under uncertainty. The research in this area remains active, which is demonstrated by a work of~\cite{sun2022optimization}. They introduce new solution concepts based on the ideas of~\cite{suijs1999cooperative}. Specifically, constraint optimization techniques are used to derive solution concepts that optimize various objective functions based on the characteristic function.

In this paper, we introduce a new solution concept called the \emph{SSD-core} for cooperative games with stochastic characteristic functions. Our approach does not require specifying exact levels of risk aversion for each player. Instead, it leverages second-order stochastic dominance to ensure that the proposed allocations are acceptable to all risk-averse players. This makes our solution concept more general and applicable to a broader range of scenarios compared to previous models. It also enhances robustness with respect to slight variations in levels of risk aversion of players.

We apply our SSD-core to the multiple newsvendors problem, a classic example in inventory management and cooperative game theory. The multiple newsvendors problem extends the single newsvendor problem by considering several individuals selling newspapers, where decisions about ordering and cooperation affect their profits under demand uncertainty. The fundamental question is: How many newspapers should the newsvendors purchase to maximize their profits while considering the risk associated with uncertain demand?
In a competitive setting, this question is not difficult to answer under reasonable assumptions. However, when cooperation among agents is introduced, the problem becomes more complex due to the various possible cooperative arrangements. The two primary approaches are either placing a single collective order or pooling inventories.

Previous studies, such as \cite{ozen2011convexity}, have analyzed the convexity and the core of the multiple newsvendors game with respect to expected profits, assuming risk-neutral players. \cite{hartman2000cores} formulated the problem as a cost game with penalties and investigated the core under various demand distributions, including symmetric and normal distributions. \cite{dror2011survey} provided broader perspectives on inventory games, including joint-replenishment and dynamic lot-sizing games. Finally, \cite{yang2021multilocation} propose four variants of cooperation in inventory games, which effectively summarize how players can collaborate through either centralized ordering or by fulfilling each other's demand (pooling inventories).

Our approach differs from mentioned ones by modeling the game as inherently random and focusing on risk-averse players. We limit our analysis to games where a centralized order is placed, and newsvendors can fulfill each other's needs, for example, by sending newspapers to one another or recommending other stores to customers. This aligns with the model proposed by \cite{ozen2011convexity} and the system described by \cite{yang2021multilocation}, referred to as \emph{CP} (\emph{centralized order with pooled inventories}).

By applying the SSD-core, we can identify fair and stable allocations acceptable to all risk-averse newsvendors, without needing to specify their exact levels of risk aversion. This approach not only broadens the application of cooperative game theory to risk-averse scenarios but also offers practical insights for inventory management under uncertainty.

In summary, our contributions are twofold: we develop a new solution concept for cooperative games with stochastic characteristic functions based on second-order stochastic dominance, and we demonstrate its applicability through the multiple newsvendors problem. This approach addresses the limitations of previous models that require specific risk parameters or utility functions and offers a universal framework for cooperative decision-making among risk-averse players.

\section{Model and stochastic dominance core}

We initiate this section by a general definition of \emph{stochastic cooperative game}, which extends the definition of a standard cooperative game. We define different types of stochastic payoffs and further formulate the \emph{SSD-core}. We conclude with auxiliary definitions and results from classical cooperative game theory necessary in our analysis of the SSD-core.

\begin{defn}[Stochastic TU-game]
    \emph{Stochastic TU-game} is a pair $(N,v)$, where $N=\{1,2,\ldots,n\}$ is a set of players and $v=(v(S))_{S\subseteq N}$ is a multivariate random variable.
\end{defn}
In this text, we restrict to one-dimensional random variables $v(S)$ for every $S \subseteq N$, similarly to~\cite{suijs1999cooperative}. The tradeoff in generality results in broader applicability, as it is often more practical in real-world scenarios to only assume marginal distributions of $v(S)$.

Players' payoffs are regarded as random vectors, however, it is often desirable to consider special types, simplifying both their analysis and representability (see~\ref{app:unstructed-allocations} for more discussion).  The following two types reflect that players agree on a proportional division of the realization of $v(N)$.
\begin{defn}[Stochastic payoffs]\label{def:types_of_payoffs}
    For a stochastic TU-game $(N,v)$, a random vector $(x_i)_{i \in N}$ is a \emph{stochastic payoff}
\begin{itemize}
        \item \emph{without transfer payments} if $\forall i \in N:$ \ $x_i=r_i \cdot v(N)$, where  $r_i\geq0$, 
        \item \emph{with transfer payments} \
 if  $\forall i \in N:$\ $x_i=d_i+r_i (v(N)-\E[v(N)])$, where $d_i\in \R$ and $ r_i\geq0$.
    \end{itemize}
    Further, these payoffs are \emph{efficient} if $d(N) = \E[v(S)]$ and $r(N) = 1$.
\end{defn}

Definition~\ref{def:types_of_payoffs} introduces vectors $r \in \mathbb{R}^n$ and $(d,r) \in \R^{2n}$, which can equally represent stochastic payoffs without and with transfer payments,. We refer to these vectors as \emph{allocations of players} when emphasizing the specific values $d_i, r_i \in \R$, or as \emph{types of allocations} when focusing on the nature of the stochastic payoffs.

Before we formally introduce the SSD-core, we recall the second order stochastic dominance and present conditions for various distributions under which stochastic dominance occurs. Results in Lemma~\ref{lemma:ssd_conditions} are straightforward and their proofs can be found in~\cite{wolfstetter1993stochastic}.
\begin{defn}[Second order stochastic dominance (SSD)]\label{def:second_stoch_dom}
       Let $X,Y$ be random variables and $F_X,F_Y$ their cumulative distribution functions. We say \emph{$X$ stochastically dominates $Y$ in second order sense}, $X\succeq_{SSD} Y$, if 
       \begin{equation*}
        \forall u\in \R : \int_{-\infty}^{u} ( F_X(z)- F_Y(z) ) \,dz\leq 0,
        \end{equation*}
    or equivalently 
    if for all concave utility functions $u$, i.e., utility functions for which $\forall x \in \R:$ $u''(x)\leq 0$, it holds
    \begin{equation}\label{def:SSD}
        \E [u(X)]\geq \E [u(Y)].
    \end{equation}
\end{defn}
\begin{lemma}[SSD conditions]\label{lemma:ssd_conditions}
    Let $X,Y$ be random variables. The following conditions are characterization of the relation  $X\succeq_{SSD} Y$ for various distributions:
    \begin{itemize}
        \item $\mu_X\geq \mu_Y$ and $\sigma_X^2\leq \sigma_Y^2$ if $X$ and $Y$ are normally distributed as $X\sim N(\mu_X,\sigma_X^2)$ and $Y\sim N(\mu_Y,\sigma_Y^2)$.
        \item $a_X \geq a_Y$ and $b_Y \leq b_X +(a_X-a_Y)$, or the latter equivalently $\E [X]\geq \E[Y]$, if $X$ and $Y$ are uniformly distributed as $X\sim U[a_X,b_X]$ and $Y\sim U[a_Y,b_Y]$.   
        \item  $k_X\cdot\theta_X\geq k_Y\cdot \theta_Y$ and $\theta_X\geq \theta_Y$ if $X$ and $Y$ are gamma distributed as $X\sim \Gamma(k_X,\theta_X)$ and $Y\sim \Gamma(k_Y,\theta_Y)$, where $k$ is the shape parameter and $\theta$ is the scale parameter.
        \item  $\forall k\in\{1,2,\ldots,K\}:\ \sum_{i=1}^kx_i\geq \sum_{i=1}^ky_i$ if $X$ and $Y$ are discretely uniformly distributed with realizations $x_1\leq x_2\leq\ldots,x_K$ and $y_1\leq y_2\leq\ldots,y_K$ and each of the realization having probability $\frac{1}{K}$.  
    \end{itemize}
\end{lemma}
These conditions compare two distributions of a same type which differ only in their parameters. In general, two sufficiently different distributions can be often incomparable by SSD, e.g., if one distribution is normal with positive variance and the second is uniform on a bounded interval. Assuming players' preferences being modelled by SSD leads to our definition of the SSD-core.

\begin{defn}[SSD-core]
    Let $(N,v)$ be a stochastic TU-game. The \emph{SSD-core} is a set of efficient stochastic payoffs $x$ denoted by $\textbf{DC}(v)$ for which it holds that
    $$\forall S\subseteq N: x(S)\succeq_{SSD}v(S) \And x(N)\text{ has the same distribution as }v(N).$$
\end{defn}

To further restrict to a specific type of allocations, we denote $\textbf{DC}^{(d,r)}(v)$ and $\DCr$ the SSD-core consisting of only stochastic payoffs with or without transfer payments, respectively. 

In the rest of this section, we recall classical TU-games and state auxiliary definitions, which we use in the analysis of the SSD-core. Recall a \emph{classical TU-game} is $(N,v)$, where $N$ is the player set and $v \colon 2^N \to \R$ with $v(\emptyset) = 0$. Further the \emph{core} $\core(v)$ of a classical TU-game $(N,v)$ is $\core(v) = \{x \in \R^n \mid x(S) \geq v(S), \forall S \subseteq N \text{ and } x(N) = v(N)\}$ and the \emph{nonnegative cost core} $\core_{cost}(v) = \{x \in \R^n_+ \mid x(S) \leq v(S) \text{ and } x(N)= v(N)\}$. In accordance with standard notation, $x(S) = \sum_{i \in S}x_i$. A classical TU-game $(N,v)$ is \emph{superadditive}, if $v(S) + v(T) \leq v(S \cup T)$ for $S,T \subseteq N, S \cap T = \emptyset$ and it is \emph{convex}, if $v(S) + v(T) \leq v(S \cap T) + v(S \cup T)$ for $S,T \subseteq N$. For a convex $(N,v)$, $\core(v) \neq \emptyset$ and in general games with nonempty cores are called \emph{balanced}. A cooperative game $(N,v_r)$ for $r \in \R^n$ is \emph{additive} if $v_r(S) = \sum_{i \in S}r_i$.
In our results, we further employ special classical TU-games. For simplicity, we denote $\mu_S = \E[v(S)]$ and $\sigma_S = \sqrt{\textbf{Var}([v(S)])}$.

\begin{defn}[Mean and deviation games]\label{def:mean_std}
With a stochastic TU-game $(N,v)$, we associate the following games:
\begin{itemize}
        \item \emph{mean game} $(N,\mu)$ defined as $\mu(S)=\mu_S$,
        \item \emph{deviation game} $(N,\dev)$ defined as  $\dev(S)=\frac{\sigma_S}{\sigma_N}$.
    \end{itemize}
\end{defn}

\begin{defn}[Lower bound game]\label{def:lower_game}
    Let $(N,v)$ be a stochastic TU-game and $v(S)\sim U[a_S,b_S],\ \forall S\subseteq N$ be uniformly distributed.
    Then the \emph{lower bound game} $(N,a)$ is classical TU-game defined as $a(S)=a_S$.
\end{defn}

Similarly to $\mu(S)$, $\mu_S$, and $a(S)$, $a_S$, we use $v(S)$ and $v_S$ interchangeably even for other characteristic functions.

\section{Analysis of the SSD-core}\label{sec:3} 
In this section, we analyze the SSD-core, focusing on the question of nonemptiness. We provide a detailed demonstration of our findings on $\DCdr$, which consists of stochastic payoffs with transfer payments. We establish connections between SSD-cores and the cores of classical TU-games. First, in Proposition~\ref{prop:SSD_ex_dr}, we show that the nonemptiness of the core of the mean game is always necessary. For the case of the normal distribution, Theorem~\ref{thm:norm_dr} proves that the nonemptiness of the SSD-core is equivalent to the nonemptiness of both the core of the mean game and the nonnegative cost core of the deviation game. In contrast, our main result in Theorem~\ref{thm:unif_dr} demonstrates that for uniform distributions, while the nonemptiness of the cores of the mean and lower bound games is necessary, it is not sufficient. An additional condition—convexity of the lower bound game—is sufficient. Further insight into the limits of nonemptiness is provided in Proposition~\ref{prop:superad_convex}.

In the remainder of this section, we provide an overview of similar results achieved for stochastic payoffs of type $r$, along with a slight generalization of type $(d, r)$, where the restriction $r \geq 0$ is relaxed.

We start with a general observation about the relationship between the SSD-core and the core of the mean game.
\begin{proposition}\label{prop:SSD_ex_dr}
For a stochastic TU-game $(N,v)$, it holds 
\begin{equation*}\DCdr \neq \emptyset \implies \core(\mu) \neq \emptyset,
\end{equation*}
where $(N,\mu)$ is the mean game of $(N,v)$.     
\end{proposition}
\begin{proof}
   If $\DCdr \neq \emptyset$, there exists a payoff vector $x$ such that $\forall S \subseteq N: x(S) \succeq v(S)$. Since the identity function $id \colon \R \to \R$ is concave, it follows from the definition of SSD that $\forall S \subseteq N: d(S) = id(\mathbb{E}[x(S)]) \geq id(\mathbb{E}[v(S)]) = \mu_S$. Because there exists a payoff vector $x$ satisfying $d(S) \geq \mu_S$ for all $S \subseteq N$, we conclude that $\core(\mu) \neq \emptyset$. 
\end{proof}

Proposition~\ref{prop:SSD_ex_dr} establishes a link between the SSD-core and the core of classical TU-games in general. In the case of normal distributions, the nonemptiness of the SSD-core is guaranteed by the additional condition of nonemptiness of the nonnegative cost core of another classical TU-game.

\begin{thm}[SSD-core under normal distribution]\label{thm:norm_dr}
    Let $(N,v)$ be a stochastic TU-game. Suppose $v(S)\sim N(\mu_S,\sigma_S^2),\ \forall S\subseteq N$ is normally distributed, where $\mu_S\in \R, \sigma_S^2\in \R_+$. Then
    \begin{equation*}
        \DCdr\neq \emptyset \iff \core(\mu)\neq \emptyset\text{ and } \core_{cost}(\dev)\neq \emptyset.
    \end{equation*}
    
\end{thm}
\begin{proof}
    For $x \in \DCdr$, we have $\E[x(S)]=d(S)$ and further $\var[x(S)]=(r(S))^2\sigma_N^2$.
    It follows from Lemma~\ref{lemma:ssd_conditions} that $\DCdr\neq \emptyset$ if and only if there is an efficient $(d,r)$ satisfying for every $S \subseteq N$:
    \begin{align}
    d(S)&\geq \mu_S, \label{Nssd_1}  \\ 
    \sigma^2_S&\geq \sigma^2_N (r(S))^2 \iff  \dfrac{\sigma_S}{\sigma_N}\geq r(S),  \ \text{if } \ \sigma^2_N>0. \label{Nssd_2}
    \end{align}
    Conditions~\eqref{Nssd_1} together with efficiency of $x$ is equivalent to $ \core(\mu)\neq \emptyset$ and Conditions~\eqref{Nssd_2} with efficiency of $x$ to $\core_{cost}(\dev)\neq \emptyset$.
\end{proof}    

In the case of the uniform distribution, the situation becomes more complex. The variance game is replaced by the lower bound game, which in turn requires to shift our focus from the nonnegative cost core of the lower bound game to its core. Moreover, the nonemptiness of both cores is no longer sufficient; an additional condition related to the lower bound game must be considered. We demonstrate that the convexity of the lower bound game provides a sufficient condition.

\begin{lemma}\label{lemma:eq_add_games}
    Let $(N,a)$ be a TU-game, $(N,v_r)$ where $r\in\R^n$ be an additive TU-game and $K\in \R$. Then $\core(a+K\cdot v_r)=\core(a)+K\cdot r$.
\end{lemma}
\begin{proof}
    First, note $\core(K\cdot v_r)=K\cdot r$. For $x \in \core(a)$, $x(S) \geq a(S)$, thus $x(S) + K \cdot r(S) \geq (a + K \cdot v_r)(S)$, therefore $x + K \cdot r \in \core(a + K \cdot v_r)$.
    Second, suppose $z\in \core(a+K\cdot v_r)$ and define $y$ as $y_i = z_i - K \cdot r_i$. As $z(S)\geq a(S) + K\cdot r(S)$ yields $y(S)\geq a(S)$, it holds $y \in \core(a)$, which concludes the proof.
\end{proof}

\begin{thm}[SSD-core under uniform distribution]\label{thm:unif_dr}
    Let $(N,v)$ be a stochastic TU-game. Suppose  $v(S)\sim U[a_S,b_S],\ \forall S\subseteq N$ is uniformly distributed, where $a_S,b_S\in \R$, $a_S<b_S$. Then the following implications hold:
    \begin{align*}
    &\DCdr\neq \emptyset \implies \core(\mu)\neq \emptyset \And \core(a)\neq \emptyset, \\
       &(N,a) \ \text{is a convex game } \And \core(\mu)\neq\emptyset\implies\DCdr\neq \emptyset,
    \end{align*}
    where $(N,\mu)$ is the mean game and $(N,a)$ is the lower bound game of $(N,v)$.
\end{thm}

\begin{proof}
A stochastic payoff $x\in \DCdr$ needs to satisfy $x(S)\succeq_{SSD} v(S)$, $\forall S\subseteq N$. To be able to apply Lemma~\ref{lemma:ssd_conditions}, we express $x(S)$ as $$x(S)\sim U[d(S)+r(S)(a_N-\mu_N),d(S)+r(S)(b_N-\mu_N)].$$
Applying the lemma, $x \in \DCdr$ if and only if it satisfies for every $S \subseteq N$:
\begin{align}
    &d(S)\geq \mu_S, \label{unif_1} \\ 
    &d(S)\geq a_S +r(S)(\mu_N-a_N). \label{unif_2}  
\end{align}
Conditions~\eqref{unif_1} and~\eqref{unif_2} can be viewed as a condition for $d$ to be in the intersection of two cores; namely inequalities in~\eqref{unif_1} translates to $d \in \core(\mu)$ and those in~\eqref{unif_2} to $d \in \core(a+(\mu_N-a_N)\cdot v_r)$, where $v_r(S) = r(S)$, or equivalently by Lemma~\ref{lemma:eq_add_games}, $d \in \core(a)+(\mu_N-a_N)r$.
Nonemptiness of $\DCdr$ is thus equivalent to existence of $d\in \core(\mu)\cap (\core(a)+(\mu_N-a_N)r)$, i.e., finding some $d \in \R^n$ in the intersection of $\core(\mu)$ and $\core(a)$ shifted by $(\mu_N-a_N)r $ for some $r \geq 0$: $r(N)=1$. It is immediate from Lemma~\ref{lemma:eq_add_games} that once either $\core(\mu) = \emptyset$ or $\core(a)= \emptyset$, we have $\DCdr = \emptyset$.

Now assume nonemptiness of both $\core(\mu)$, $\core(a)$ and convexity of $(N,a)$. We show that for any $d \in \core(\mu)$, we can find $x \in \core(a)$ and $r$ satisfying $r \geq 0$, $r(N)=1$ such that $d \in \core(\mu) \cap\left( \core(a) + (\mu_N - a_N)r\right)$. For this purpose, we employ the following iterative process $\mathcal{P}$:
\begin{enumerate}
    \item Set $x^0=d$. 
    \item At step $m$, select player $k_m$ defined as: 
    $$ k_m=\min\{k\in N: a_S<x^{m-1}(S), \ \forall S\subseteq N,\ k\in S\}.$$
    \item Set $x^m=x^{m-1}- t_m e_{k_m}$, where $e_{k_i}$ is $k_i$-th vector of canonical basis in $\R^n$ and $t_m$ is a length of step $m$ given by 
    $$t_m=\min_{S\subseteq N, S\ni k_m}[x^{m-1}(S)-a_S].$$
    We denote $\mathcal{S}_m=\arg \min_{S\subseteq N, S\ni k_m} [x^{m-1}(S)-a_S].$
    
    \item If $x^{m}(N)\neq a_N$ then go to step $m+1$. If $x^m(N)=a_N$, set $x = x_m$ and stop the process. We denote the final step as $m_{end}.$ 
\end{enumerate}
 Notice that at each step $m$, $x^m(S) \geq a(S)$ for every $S \subsetneq N$. For $x^{m_{end}} \in \core(a)$, it remains to verify that $x^{m_{end}}(N)=a_N$. A key observation is that $\mathcal{S}_m$ is closed under union, or equivalently, that $\bigcup_{S\in \mathcal{S}_m}S$ is equal to the inclusion-wise maximal element in $\mathcal{S}_m$. To see this, suppose $S,\ T\in \mathcal{S}_m$, therefore, $a_S=x^m(S)$ and $a_T=x^m(T)$. From convexity of $(N,a)$, it follows that $a_S+a_T\leq a_{S\cup T}+a_{S\cap T}$. In combination with $a_{S\cup T}\leq x(S\cup T)$ and $ a_{S\cap T}\leq x(S \cap T)$, it follows that $a_{S\cup T}= x(S\cup T)$, thus $S \cup T \in \mathcal{S}_m$. We denote by $S_m$ the inclusion-wise maximal set in $\mathcal{S}_m$.
 
 Let us follow the process $\mathcal{P}$ to verify $x^{m_{end}}(N) = a(N)$ by showing for every $m$ that $x^m(S_1 \cup \dots S_m) = a(S_1 \cup \dots \cup S_m)$.
\begin{itemize}
    \item At step $m=1$: From the definition of $\mathcal{P}$, $x^1(S_1)=a_{S_1}$ is immediate.
    \item At step $m=2$: From $x^2(S_2)=a_{S_2}$ and $x^2(S_1)=x^1(S_1) = a_{S_1}$, and convexity of $(N,a)$, we have
    $$x^2(S_1 \cup S_2)+x^2(S_1\cap S_2)=a_{S_1}+a_{S_2} \leq a_{S_1\cup S_2}+a_{S_1 \cap S_2}$$ 
    and since $a_S \leq x^2(S)$, $\forall S\subseteq N$, we have  
    $$x^2(S_1 \cup S_2)+x^2(S_1\cap S_2)=a_{S_1\cup S_2}+a_{S_1 \cap S_2}.$$ 
    Specifically $x^2(S_1\cap S_2)=a_{S_1 \cap S_2}$ and $x^2(S_1 \cup S_2)=a_{S_1\cup S_2}$.
    
    \item At step $m>2$: The general step follows the argument of $m=2$. We have $x^m(S_m)=a_{S_m}$ and $x^{m-1}(S_1 \cup \dots \cup S_{m-1}) = a_{S_1 \cup \dots \cup S_{m-1}}$.
     Therefore
        $$x^m(S_1 \cup\ldots \cup S_{m} )+x^m((S_1 \cup \ldots \cup S_{m-1})\cap S_m)=a_{(S_1 \cup \ldots \cup S_{m-1})}+ a_{S_{m}}, $$
        and at the same time from the convexity of $(N,a)$, 
      $$a_{(S_1 \cup \ldots \cup S_{m-1})}+ a_{S_{m}} \leq a_{((S_1 \cup \ldots S_{m-1}) \cap S_{m})}+a_{(S_1 \cup \ldots \cup S_{m})}.$$
    Similar to $m=2$, we get $x^m(S_1 \cup \ldots \cup S_{m} )=a_{(S_1 \cup \ldots \cup S_{m})}.$
    \item It remains to show that $S_1\cup \dots \cup S_{m_{end}} = N$. If not, then there is $i\in N\setminus \bigcup_{i=1}^{m_{end}}S_i$ for which it holds that every coalition $S$, $i\in S$ satisfies $x^{m_{end}}(S)\geq \mu_S>a_S$ since the distribution of $v(S)$ is not degenerated, i.e., $a_S<b_S$, thus $m_{end}$ is not the final step. 
    
\end{itemize}
    
\end{proof}
The characterization of $\DCdr\neq \emptyset$ in a form of conditions from Theorem~\ref{thm:unif_dr} falls somewhere between balancedness and convexity of $(N,a)$. In what follows, we offer further insight into the extent of this gap. Example~\ref{example:super-example} shows that the superadditivity of a balanced $(N, a)$ is neither sufficient nor prohibitive for the nonemptiness of the SSD-core.

\begin{example}\label{example:super-example}
    Let $(N,a)$ be defined as $a_{12} = a_{23} = a_N = 3$ and $a_S = 0$, otherwise. This game is superadditive, however, the convexity is violated for $a_{12}+a_{23}>a_{123}+a_2$. Now let $(N,\mu)$ be defined as $\mu_1=\mu_3=5$, $\mu_2=2$, $\mu_{N}=12$, and $\mu_S=a_S$, otherwise.
    It holds $\core(a)=\{(0,3,0)\}$ and $\core(\mu)=\{(5,2,5)\}$. However, there is no $d \in \R^n$ and $r \in \R^n_+$ satisfying $r(N) = 1$ and $d\in \core(\mu)\cap (\core(a)+r(\mu_N-a_N))$, therefore $\DCdr = \emptyset$. However, if we change $\mu_2=5$, and $\mu_N=15$, we have $\core(\mu) = \{(5,5,5)\}$ thus $d = (5,5,5)$ and $r=\frac{1}{12}(5,2,5)$ satisfy the condition.
\end{example}

Example~\ref{example:super-example} indicates that the characterization does not necessarily depend just on properties of $(N,a)$ but rather on a relation between $(N,a)$ and $(N,\mu)$. It is thus much more interesting that simply convexity of $(N,a)$ is sufficient in combination with balancedness of $(N,\mu)$. The relation is not symmetric though, as it is easy to find an example where balancedness of $(N,a)$ and convexity of $(N,\mu)$ is not sufficient.

Understanding the importance of convexity of $(N,a)$ might be viewed from the point of view of the \emph{Weber set}, (\cite{Weber1988}). This solution concept, denoted by $W$, satisfies for every game that $\core(v) \subseteq W(v)$ and $\core(v) = W(v)$ if and only if $(N,v)$ is convex. In this sense, convex games are characterized by having the largest possible core, as the core coincides with the entire Weber set. The Weber set is a convex set, which is tightly bounded by hyperplanes $\mathcal{H}_i = \{x \in \R^n \mid x_i = v_i\}$ for $i \in N$. The following proposition illustrates that if a game is not convex, it is enough for just one player $i \in N$ to have the intersection of their hyperplane $\mathcal{H}_i$ with the Weber set $W(v)$ lying outside the core ($\mathcal{H}_i \cap W(v) \notin \core(v)$) to violate sufficiency. In other words, even if the Weber set closely aligns with the core, a small portion of its boundary missing -- particularly at the hyperplane of a single player -- is enough to break the condition for nonemptiness of $\DCdr$. This highlights how even a slightly weaker property of $(N,a)$ than convexity might not be enough.

\begin{proposition}\label{prop:superad_convex}
    Let $(N,a)$ be a balanced classical TU-game satisfying
    \begin{equation*}
        \exists i \in N,\ \forall x \in \core(a): x_i > a_i.
    \end{equation*}
    Then there are $b_S \in \R,\ \forall S \subseteq N$ such that $\DCdr = \emptyset$ for $(N,v)$ defined as $v(S) \sim U[a_S,b_S]$.
\end{proposition}
\begin{proof}
    We construct a mean game $(N, \mu)$ by setting $\mu_i = a_i + \varepsilon$, where $\varepsilon > 0$ is sufficiently small to ensure that $\mu_i < x_i$ for every $x \in \core(a)$. Further, setting $\mu_{N \setminus i }=\mu_N-\mu_i$ ensures for every $y \in \core(\mu): y_i = \mu_i$ and finally, by setting $\mu_S > a_S$ for the rest of coalitions together with Theorem~\ref{thm:unif_dr} implies $\DCdr = \emptyset$.
\end{proof}    

\subsection{Overview of results for other types of payoffs}
The analysis of stochastic payoffs with transfer payments can be generalized to $r \in \mathbb{R}^n$, lifting the condition of nonnegativity. This idea was proposed in~\cite{sun2022optimization} and effectively allows for negative correlation between $v(N)$ and the payoff $x_i$, i.e., $\text{corr}(x_i, v(N)) = -1$. An agent may prefer a negative $r_i$ if they are pessimistic about the outcome and wish to \emph{distance} themselves from a negative result (realization is less than the expected value), or more specifically, profit from it, even at the cost of suffering penalties if a positive outcome occurs. We present results regarding the uniform and normal distribution of these allocations types, which we denote by $(d,r_\pm)$, to compare with the results in Theorems~\ref{thm:norm_dr} and~\ref{thm:unif_dr}. As the proofs differ only in minor details, we omit them.
\begin{proposition}
    Let $(N,v)$ be a stochastic TU-game, where $v(S)\sim N(\mu_S,\sigma_S^2)$, $\forall S \subseteq N$ is normally distributed with parameters $\mu_S\in \R, \sigma_S^2\in \R_+$. It holds
    $\DCdrpm\neq \emptyset$ if and only if there are $d \in \R^n$, $r \in \R^n$:
    \begin{itemize}
        \item $d(S)\geq \mu_S,\ \forall S\subseteq N$ and $d(N)=\mu_N,$
        \item $|r(S)|\leq \frac{\sigma_S}{\sigma_N},\ \forall S\subseteq N$ and $r(N)=1.$
    \end{itemize}
\end{proposition}
Unfortunately, conditions of the second type do not immediately translate to nonemptiness of any core of $(N,\dev)$ due to the absolute value. The result concerning uniform distributions (Proposition~\ref{prop:unif_general_type}) is more interesting because the characterization of $\DCdrpm\neq\emptyset$ simplifies to $\core(a)\neq\emptyset$ and $\core(\mu)\neq\emptyset$. The reason for the simplification is quite straightforward; for any $d \in \core(\mu)$ and any $x \in \core(a)$, one can construct $r \in \R^n$ as $r = d-x$ as the entries of $r$ are allowed to be negative.
\begin{proposition}\label{prop:unif_general_type}
    Let $(N,v)$ be a stochastic TU-game, where $v(S)\sim U[a_S,b_S]$,  $\forall S\subseteq N$ is uniformly distributed with $a_S,b_S\in \R$, $a_S<b_S$. Then \begin{equation*}
        \DCdrpm\neq \emptyset \iff \core(a)\neq \emptyset\ \& \ \core(\mu) \neq \emptyset.
    \end{equation*}
\end{proposition}

We conclude this section with an overview of results for stochastic payoffs without transfer payments. This type of allocations is more restrictive then those with transfer payments as each $x$ expressed by $r$ can be expressed by $(d,r)$ where $d_i=r_i \cdot \E[v(N)]$. On the other hand, these payoffs yield simpler analysis and can accompany more distributions since we need only scale family of distributions, i.e., distributions not changing under multiplication by a constant. This is reflected in Theorem~\ref{proposition:r_type_all_distr}, where we provide additional results for gamma and discrete uniform distributions. We will also see in the following section that, in the application to the multiple newsvendor problem, they provide a more easily explainable interpretation compared to payoffs with transfer payments, where the interpretation might not be as clear.
 
\begin{thm}[SSD-core conditions]\label{proposition:r_type_all_distr}
Let $(N,v)$ be a stochastic TU-game. 
An efficient stochastic vector without transfer payments $x = r \cdot v(N)$ lies in the SSD-core $\DCr$ if and only if the following conditions are met:

\begin{itemize}
    \item If $v(S)\sim N(\mu_S,\sigma_S^2),\ \forall S\subseteq N$ is normally distributed with parameters $\mu_S\in \R, \sigma_S^2\in \R_+$ then
    $$r(S)\geq \frac{\mu_S}{\mu_N} \And r(S)\leq \frac{\sigma_S}{\sigma_N}, \ \forall S\subseteq N. $$
    \item If $v(S)\sim U[a_S,b_S], \ \forall S\subseteq N  $ is uniformly distributed with $a_S,b_S\in \R$, $a_S<b_S$ then 
    $$r(S)\geq \max\{\frac{\mu_S}{\mu_N},\frac{a_S}{a_N}\}, \ \forall S\subseteq N. $$
    \item If $v(S)$ has a discrete uniform distribution with equiprobable realizations $\omega_1\leq \omega_2\leq\ldots\leq \omega_T$ $,\ \forall S\subseteq N$, then 
    $$\forall S\subset N: r(S)\geq \max_{k\in \{1,\ldots, T\}} \dfrac{\sum_{i=1}^k v(S,\omega_i)}{\sum_{i=1}^k v(N,\omega_i)}, \ \forall S\subseteq N.$$
    \item If $v(S)\sim \Gamma(k_S,\theta_S),\ \forall S\subseteq N$ has a gamma distribution with $k_S,\theta_S\in (0,\infty)$ begin the shape and scale parameters.
    $$r(S)\geq \dfrac{k_S\cdot \theta_S}{k_N\cdot \theta_N} \And r(S)\geq \dfrac{\theta_S}{\theta_N}.$$
\end{itemize}
\end{thm}
\begin{proof}
    The proofs for the normal and uniform distributions are similar to those of Theorems~\ref{thm:norm_dr} and~\ref{thm:unif_dr}, thus we omit them. For the discrete uniform distribution, the criteria for SSD are based on partial sums from Lemma~\ref{lemma:ssd_conditions} of the ordered realizations of the random variable. Hence, $x(S)\succeq v(S)$ if
\begin{equation*}
    r(S)\sum_{i=1}^k v(N,\omega_i)\geq \sum_{i=1}^k v(S,\omega_i),\ k={1,2,\ldots,T}.
\end{equation*}
    For $v(S)$ following the gamma distribution, $x(S)$ can be expressed as $x(S)\sim \Gamma(k_S,r(S)\cdot \theta_S)$. The conditions for the $\DCr\neq \emptyset$ follow directly from Lemma~\ref{lemma:ssd_conditions}.
\end{proof}

The conditions in Theorem~\ref{proposition:r_type_all_distr} are given in the form of fractions. If the denominator is zero, the conditions translate to the numerator of the fraction being less than or equal to $0$. For gamma distribution, this problem does not arise as $\Theta_N > 0$.

\section{Multiple risk-averse newsvendors}\label{sec:risk_av_news}
In this section, we study risk-averse behavior of players in the multiple newsvendors problem. We restrict our analysis to a single-period setting, i.e., only one order is placed. We follow the model outlined in introduction, which involves a centralized order decision (one order is placed) for the pooled inventory of players (players can satisfy demand of each other). To model risk-averse behaviour, we use the second-order stochastic dominance, which was not, to the best of our knowledge, considered in the literature in the case of multiple newsvendors problem. We apply results from the previous section on nonemptiness of the SSD-core and interpret these in the terms of the newsvendors problem.

A \emph{stochastic multiple newsvendors problem} is $((Y_S)_{S \subseteq N},c,p)$ where
\begin{align*}
    Y_S&\ldots \text{a random demand of a coalition $S$ for a single period}, \\
    c&\ldots \text{unit purchasing price for players}, \\
    p&\ldots \text{unit selling price}
\end{align*}
with $0 < c < p$.
In addition, an implicit function describing the profit of a coalition $S$, for a given parameter $q_S \in \R$, which represents the quantity of ordered units for coalition $S$, is assumed:

$$v(S,q_S)=p\cdot\min(Y_S,q_S)-c\cdot q_S.$$
This construction yields that $v(S,q)$ is not only random but also depends on parameter $q$. To obtain a stochastic TU-game, we set $v(S) = v(S,q^*)$ with the optimal value of order $q^*$ where $q_S^*=\arg \max_{q_S\in \R} \E[v(S,q_S)]$.
Value $v(S)$ describes the random profit of a coalition $S$ under the optimal order quantity $q_S^*$ which is derived under expectation.

\begin{defn}[Stochastic multiple newsvendors game]\label{def:multiple_news_risk_av}
Stochastic TU-game $(N,v)$ is a \emph{stochastic multiple newsvendors game of $((Y_S)_{S \subseteq N},c,p)$} if it is defined as
\begin{align*}
    v(S)&=p\cdot\min(Y_S,q^*_S) -c\cdot q^*_S,
\end{align*}    
where  $q_S^*=\arg \max_{q_S\in \R} \E[v(S,q_S)]$.
\end{defn}

Stochastic multiple newsvendors game is more general than the one in~\cite{ozen2011convexity}. The main difference is that we use \emph{random} characteristic function with optimal value of the parameter $q^*_S$, while they use \emph{deterministic} characteristic function with value of $S$ being equal to $\max_{q_S\in \R}\E[v(S,q_S)]$. This means our stochastic game still contains all the information about the demand distribution, even though it is only under optimal ordering based on expectation. Further, it can accommodate various risk approaches of the players, compared to~\cite{ozen2011convexity}, where all players are assumed to be risk-neutral.

In practice, Definition~\ref{def:multiple_news_risk_av} covers several scenarios. For instance, newsvendors order newspaper from a specified firm, and the unit purchasing price $c$ already includes the company's transportation cost for transferring newspaper from one newsvendor to another. In another scenario, the transportation costs may be negligible, or it might be assumed that there are none. One last scenario might be to assume that if a newsvendor cannot meet the demand, he can refer the customer to another newsvendor with whom he cooperates.

Our construction is mathematically sound as under SSD, the choice of $q^*_S$ for defining $v(S)$ is reasonable, since no $v(S,q_S)$ dominates $v(S,q_S^*)$ when $q_S\neq q_S^*$ if the maximum $v(S,q_S^*)$ is unique. This is due to the fact that under SSD a random variable $X$ cannot be dominated by a random variable $Y$ if $\E[X]>\E[Y]$. 

The question of cooperation among all risk-averse newsvendors translates to the nonemptiness of the SSD-core. Typically, random demand is assumed to follow a discrete distribution, reflecting the fundamental nature of inventory in discrete units. However, for large demand quantities, the computational complexity of handling discrete distributions grows significantly, making the problem increasingly intractable as the state space expands. Approximating the problem with continuous distributions presents a powerful alternative. Continuous models not only simplify the computational burden but also provide a smoother mathematical framework, facilitating more tractable solutions. In practical terms, continuous distributions capture real-world demand more flexibly, especially in industries where demand does not naturally divide into discrete units, such as energy or commodities measured in large quantities. Furthermore, deriving conditions for the nonemptiness of the SSD-core becomes less demanding with continuous distributions. This is because continuous formulations allow for more general conditions and avoid the combinatorial complexity inherent in discrete models. For these reasons, we restrict our analysis to continuous distributions.

\subsection{Cooperation under uniform distribution}
In the rest of this section, we focus on $Y_S \sim U[a_S,b_S]$ following uniform distribution and interpret possibility of cooperation of all newsvendors by analysing nonemptiness of the SSD-core. We assume the newsvendors agree on stochastic payoffs without transfer payments. In our analysis, we are unable to directly apply Proposition~\ref{proposition:r_type_all_distr}. Nevertheless, we can derive the distribution of $v(S)$ in terms of the demand distribution $Y_S$ for any subset $S\subseteq N$. Once this distribution is established, we can proceed by following the same proof structure as outlined in the proposition, adjusting for the derived form of $v(S)$.
From the definition of $q_S^*$, one can derive $q_S^*=F_{Y_S}^{-1}(\frac{p-c}{p})$ and specifically for $Y_S \sim U[a_S,b_S]$, we get $q_S^*=a_S+\frac{p-c}{p}(b_S-a_S)$. 
The distribution of $v(S)$ is a combination of a continuous uniform and a discrete distribution and can be described using the outcomes $\omega$ of distribution $Y_S$ as follows:
\begin{equation}\label{eq:news_char}
    v(S,\omega)=\begin{cases}  p\cdot \omega-c\cdot q_S^*  & \text{if} \ \omega\in [a_S,q_S^*)  \\
                                 (p-c)q^*_S          & \text{if} \ \omega\in [q_S^*,b_S]  %
        \end{cases}, 
 \end{equation}
The cumulative distribution function has the form of the uniform distribution up to $F_{Y_S}^{-1}\left(\frac{p-c}{p}\right)$, at which point it exhibits a jump to $1$.  For further clarification, we provide Figure~\ref{fig:alpha_cut}, which visually illustrates the shape of a cumulative distribution function with this structure. We refer to this type of distribution as an \emph{$\alpha$-cut uniform distribution}.

\begin{figure}
    \centering
    \includegraphics[scale=0.8]{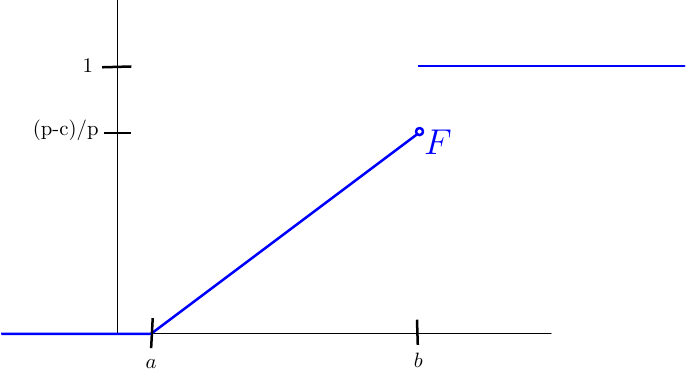}
    \caption{$(\frac{p-c}{p})$-cut uniform distributions.}
    \label{fig:alpha_cut}
\end{figure}

\begin{defn}[$\alpha$-cut uniform distribution]
Let $a_Z,b_Z\in \R$, where $a_Z<b_Z$ and $\alpha \in (0,1)$. A random variable $Z \sim U_\alpha[a_Z,b_Z]$ follows \emph{$\alpha$-cut uniform distribution} when it has the following cumulative distribution function:
\begin{equation}\label{eq:news_char_cdf}
    F_Z(x)=\begin{cases} 0  & \text{if} \ x<a_Z,  \\
    \frac{x-a_Z}{b_Z-a_Z}\cdot \alpha  & \text{if} \ x\in[a_Z,b_Z),  \\
    1         & \text{if} \ x\geq b_Z.
        \end{cases} 
 \end{equation}
\end{defn}
In the following lemma, we derive the SSD conditions for two $\alpha$-cut uniform distributions with the same $\alpha$.

\begin{lemma}[SSD condition for $\alpha$-cut uniform distribution]\label{lemma:alpha_cut}
Let $X$ and $Y$ be random variables both possessing an $\alpha$-cut uniform distribution for the same $\alpha\in(0,1)$ and let $a_X,a_Y,b_X$, and $b_Y$ be the corresponding parameters. Then $X\succeq_{SSD} Y$ if an only if 
$$a_X\geq a_Y \And (2-\alpha)\cdot b_X+\alpha\cdot a_X\geq (2-\alpha)\cdot b_Y+\alpha\cdot a_Y.$$    
\end{lemma}

\begin{proof}
    To derive the conditions, we use the formulation of SSD from Definition~\ref{def:second_stoch_dom} concerning cumulative distribution function, i.e., 
    \begin{equation}\label{eq:SSD_lemma}
        \forall u\in \R : I(u)=\int_{-\infty}^{u} ( F_X(z)- F_Y(z) ) \,dz\leq 0.
        \end{equation}
        We can easily see that for any $u\leq \min\{a_X,a_Y\}$, $I(u)=0$, thus we only need to analyze situation, where $u>\min\{a_X,a_Y\}$. We can further see that if $a_X<u<a_Y$ then $I(u)>0$, which means $X$ cannot dominate $Y$. Hence, a necessary condition for $X$ to dominate $Y$ is $a_X\geq a_Y$. For $a_X\geq a_Y$, we derive further conditions for $X$ to dominate $Y$. We can simply see that for $b_X\geq b_Y$, $F_X(u)\leq F_Y(u)$ for every $u \in \R$, thus, the variable $X$ dominates $Y$ (actually, in the first order stochastic dominance). Let the relation between $b_X$ and $b_Y$ be $b_X\leq b_Y$. Then the distribution function $F_X$ and $F_Y$ intersect on interval $(a_Y,b_Y)$ at a point denoted by $h$. The crucial observation is that $I(u)$ is decreasing on interval $(a_Y,h)$ and increasing on $(h,b_Y)$. This enables us to calculate the other condition for dominance just as $I(b_Y)\leq 0$ because if $I(b_Y)\leq 0$ then $I(u)\leq 0,\ \forall u\in\R$.
        It remains to calculate $I(b_Y)$:
          \begin{align*}
        I(b_Y)&=\int_{-\infty}^{b_Y} (F_X(z)- F_Y(z) ) \,dz=(b_Y-b_X)+\alpha\cdot \frac{b_X-a_X}{2}-\alpha\cdot\frac{b_Y-a_Y}{2}\\
        &= -b_X\cdot(1-\frac{\alpha}{2})-\frac{\alpha}{2}\cdot a_X+ (1-\frac{\alpha}{2})\cdot b_Y+\frac{\alpha}{2}\cdot a_Y
        \end{align*}
        Therefore, the second condition, which together with $a_X\geq a_Y$ makes the conditions for $X\succeq_{SSD} Y$, is as follows $$b_X\cdot(2-\alpha)+\alpha\cdot a_X \geq(2-\alpha)\cdot b_Y+\alpha\cdot a_Y.$$  
        Notice, that the condition works for $b_X>b_Y$. 
\end{proof}
\begin{rem}
    For $\alpha=1$, the SSD conditions for the $\alpha$-cut uniform distribution reduce to the conditions for a uniform distribution, where the lower and upper bounds, $a_X$ and $b_X$, correspond to those in Lemma~\ref{lemma:ssd_conditions}.
\end{rem}
Now that we have conditions for two $\alpha$-cut uniform distributions, we can  derive conditions for the payoffs from the SSD-core of a stochastic multiple newsvendors game.
\begin{thm}\label{thm:newsvendor_r}
Let $(N,v)$ be a stochastic multiple newsvendors game of $((Y_S)_{S \subseteq N},c,p)$, where $Y_S\sim U[a_S,b_S]$, $\forall S\subseteq N$ is uniformly distributed with parameters $a_S,b_S\in [0,\infty)$, $a_S<b_S$. Then the SSD-core $\DCr\neq \emptyset$ if and only if 
\begin{align*}
r(S)(a_N\cdot p-(b_N-a_N)\cdot c)\geq (a_S\cdot p-(b_S-a_S)\cdot c) \And \\
r(S)(a_N\cdot (p+c)+b_N(p-c))\geq a_S\cdot (p+c)+b_S\cdot(p-c).
\end{align*}
\end{thm}
\begin{proof}
    We denote the parameters of $v(S)$ 
    as $\aS$, $\bS$ to distinguish them from parameters of $Y_S \sim U[a_S,b_S]$. For the distribution of $v(S)$, as described in~\eqref{eq:news_char}, the lower bound $\aS$ is obtained for $\omega=a_S$, i.e., $\aS=p\cdot a_S-c\cdot q_S^*$  and the upper bound is obtained at point $\omega=q_S^*$ where the cumulative distribution function is not continuous, i.e., $\bS = (p-c)\cdot q_S^*$. The distribution of $v(S)$ can be thus expressed as
    \begin{equation*}
    v(S)\sim U_\alpha[p\cdot a_S-c\cdot q_S^*,(p-c)\cdot q_S^*].
    \end{equation*}
    Further, the distribution of $x(S) = r(S) \cdot v(N)$ can be expressed in terms of $\aS$, $\bS$ as $x(S) \sim U_{\alpha}[r(S)\cdot \aN, r(S)\cdot \bN]$ or equivalently as
    \begin{equation*}
        x(S) \sim U_{\alpha}[r(S) \left( p\cdot a_N-c\cdot q_N^*\right), r(S)\cdot (p-c)\cdot q_N^*].
    \end{equation*}
    We employ Lemma~\ref{lemma:alpha_cut} to reformulate $x(S)\succeq_{SSD} v(S)$ using two conditions. The first condition from the lemma is of form
    $$r(S)(p\cdot a_N - c \cdot q_N^*)\geq p\cdot a_S - c \cdot q_S^*,$$
    which can be rewritten using $q_S^*=a_S+(b_S-a_S)(\frac{p-c}{p})$ as 
$$r(S)\cdot\left(a_N(p-c+\frac{p-c}{p}\cdot c)-b_N\frac{p-c}{p}\cdot c\right)\geq a_S(p-c+\frac{p-c}{p}\cdot c)-b_S\frac{p-c}{p}\cdot c$$
and further simplified to
$$r(S)(a_N\cdot p-(b_N-a_N)\cdot c)\geq (a_S\cdot p-(b_S-a_S)\cdot c). $$

The second condition of Lemma~\ref{lemma:alpha_cut} can be expressed as
\begin{align*}
  &r(S) \left(\frac{p+c}{p}\cdot(p-c)\cdot q_N^*+(a_N\cdot p-c\cdot q_N^*) \cdot\frac{p-c}{p}\right)\geq \\
  &\frac{p+c}{p}\cdot(p-c)\cdot q_S^*+(a_S\cdot p-c\cdot q_S^*)\cdot\frac{p-c}{p}.  
\end{align*}
This can be further simplified to:
$$q_S^*+a_S\geq r(S)(q_N^*+a_N).$$
By plugging in the optimal value $q_S^*=a_S+(b_S-a_S)(\frac{p-c}{p})$ we obtain:
$$r(S)(a_N\frac{p+c}{p}+b_N\frac{p-c}{p})\geq a_S\frac{p+c}{p}+b_S\frac{p-c}{p}.$$
    
\end{proof}

Conditions from Theorem~\ref{thm:newsvendor_r} have an interesting interpretation in terms of the model. The first condition,
$$r(S)(a_N\cdot p-(b_N-a_N)\cdot c)\geq (a_S\cdot p-(b_S-a_S)\cdot c),$$
can be viewed as \emph{protection against demand fluctuation}. Both left and right hand sides of the condition can be interpreted as the gross profit of coalition $S$ if the lowest demand is realized, minus the loss in the worst case, i.e., when $b_S$ was ordered and only $a_S$ realized. Under $r$, the newsvendors are better off in $N$. Players in $S$ will want to cooperate in the grand coalition if the protection against demand fluctuation they can guarantee for themselves is less than or equal to the share of the grand coalition’s protection against demand fluctuation that they can grab. 

The second condition
$$r(S)(a_N\cdot (p+c)+b_N(p-c))\geq a_S\cdot (p+c)+b_S\cdot(p-c)$$
is a bit more intricate to interpret. If we rewrite it in the following form,
$$r(S)\left(p\cdot\frac{a_N+b_N}{2}-c\cdot \frac{b_N-a_N}{2} \right)\geq p\cdot\frac{a_S+b_S}{2}-c\cdot \frac{b_S-a_S}{2},$$
one can view $p\cdot (a_S+b_S)/2$ as the expected net income of coalition $S$ and $c\cdot (b_S-a_S)/2$ as potential expected loss from choosing $q_S^*$ instead of $b_S$. These two terms collectively represent the \emph{market quality within coalition $S$}. Consequently, players within $S$ want to cooperate within the grand coalition, if their portion of market quality in $N$ is at least as good as the market quality within the coalition $S$. The SSD-dominating conditions thus suggest that such market quality is important for risk averse players when considering cooperation in the stochastic multiple newsvendors game.

To summarize, risk averse newsvendors ask these questions when deciding whether to cooperate with other newsvendors in the already thoroughly explained situation:
\begin{itemize}
    \item Is the portion of the protection against demand fluctuation in $N$ at least as high as  within $S$?
    \item Is the quality of the market within $S$, specifically, its expected net profit and expected loss from buying the optimal quantity $q_S^*$ instead of the maximal quantity $b_S$, at most as good as the portion of the quality of the market within $N$?
\end{itemize}

We conclude by noting that we derived a similar analysis for stochastic payoffs with transfer payments as in Theorem~\ref{thm:unif_dr}. However, the results did not appear to offer meaningful insights in the context of the model, and due to the extensive calculations involved, we decided not to include them here.

\section{Conclusion}
In this paper, we addressed the challenge of modeling risk-averse behavior in cooperative games under uncertainty by introducing the concept of the SSD-core. This approach overcomes the restrictions and impracticalities of traditional models that often require specifying exact levels of risk aversion or specific utility functions for each player. By leveraging second-order stochastic dominance, our SSD-core provides a robust solution concept that is acceptable to all risk-averse players without the need to specify their exact risk preferences.

Our main contribution lies in establishing connections between the SSD-core and the cores of associated deterministic cooperative games, which enables us to derive conditions under which players cooperation is possible. We illustrated our methods using a special allocation type under normal and uniform distributions and provided an overview of similar results for other allocation types and different distributions. We then applied our SSD-core framework to the multiple newsvendors problem, a classic example in inventory management and cooperative game theory. By modeling the game as inherently random and focusing on risk-averse players, we identified fair and stable allocations acceptable to all newsvendors without needing to specify their exact levels of risk aversion. Our approach differs from traditional models that rely on expected profits under optimal ordering, offering new insights into how cooperation can be beneficial under demand uncertainty. We provided a characterization of risk-averse behavior of players with interpretations in terms of the model, highlighting the practical applicability of our theoretical findings.

While the SSD-core provides a robust framework, it has limitations. It cannot account for covariances between coalitions, potentially missing key interdependencies.  It is often impossible to compare two different distributions, e.g. uniform and normal. Moreover, SSD does not always yield a clear ordering, limiting the SSD-core’s applicability. In~\ref{app:undominated-ssd}, we explored a generalization of SSD-core to non-comparable distributions, however, we illustrate that it may lead to irrational outcomes even if we compare normal distributions. Additionally, using the allocation types we rely mostly on scale-location family distributions which restricts generality of the approach.

Our theoretical results can be easily augmented for any distribution within scale-location or scale family depending on the allocation type. A direct research direction in this area is exploring how different allocation types influence the volume of the SSD-core.
Our results can likely be extended to other demand distributions with bounded support, including discrete distributions with finite outcomes, but may not apply to unbounded distributions like the normal.

Our findings in the multiple newsvendors problem highlight the potential for further exploration and application of the SSD-core framework. While we focused on a specific allocation type and distribution to demonstrate the applicability, there is no specific reason why other settings could not be investigated. For instance, extending our approach to bounded demand distributions, such as beta distribution or any discrete distributions with a finite number of realizations, could provide valuable insights or even similarly direct interpretation as we obtained for uniform demand distribution. Additionally, incorporating factors such as transshipment costs or other operational variables could further enhance the practical relevance of the model.


\section*{Acknowledgements}
The project was supported by the Czech Science Foundation grant no. P403-22-11117S.
The second author was further supported by the Charles University Grant Agency (GAUK 206523) and the Charles University project UNCE 24/SCI/008. 

 \bibliographystyle{elsarticle-harv} 
 \bibliography{cas-refs}


\appendix

\section{Unstructured allocation}\label{app:unstructed-allocations}
In our paper, we focused on several allocation types of stochastic payoffs. Here, we discuss the general setting, i.e., stochastic payoff $x$ as a multivariate random variable without further restrictions on its covariance structure. We generalize the result from Theorem~\ref{thm:norm_dr}, which allows to view allocations with transfer payments $(d,r)$ and the generalized $(d,r_\pm)$ as a reasonable framework for payoff allocations. 

\begin{thm}[Unstructured allocation under normal distribution]\label{thm:norm_random}
    Let $(N,v)$ be a stochastic TU-game, where $v(S)\sim N(\mu_S,\sigma_S^2), \ \forall S\subseteq N$ is normally distributed with parameters $\mu_S\in \R, \sigma_S^2\in \R_+$.
    We assume that $x$ has a multivariate normal distribution $x\sim N_n(\overline{\mu},\Sigma)$ with $\overline{\mu}\in \R^n$, $\Sigma\in \R^{n\times n}$ and $\textbf{Var}(x_i)=\Sigma_{ii}=\overline{\sigma}^2_i$. Then $\textbf{DC}(v)\neq \emptyset$ $\iff$
    $\core(\mu)\neq \emptyset$ $\&$ $\core(\sigma^2)\neq \emptyset$.
    Specifically, $x\in \textbf{DC}(v)$ if and only if it satisfies the following conditions:
    \begin{align*}
       \forall S\subseteq N \ : \mu_S &\leq \sum_{i\in S}\overline{\mu}_i , \\
        \forall S\subseteq N \ : \sigma_S^2 &\geq \sum_{i,j\in S} \rho_{i,j}\overline{\sigma}_i \overline{\sigma}_j=\Sigma_{ij}.         
    \end{align*}
\end{thm}
\begin{proof}
    We can proceed similarly as in the proof of Theorem~\ref{thm:norm_dr}. We just need to work straight with the variance instead of standard deviation game. Here, we compute only the variance of $x(S),\  S\subseteq N$ to get the second inequality:
    \begin{equation*}
        \textbf{Var}[x(S)]=\sum_{i,j\in S}\textbf{cov}(x_i,x_j)=\sum_{i,j\in S}\rho_{i,j}\overline{\sigma}_i\overline{\sigma}_j.
    \end{equation*}
\end{proof}
Consider two stochastic TU-games, the first denoted by $(N,z)$ is the one defined in Theorem~\ref{thm:norm_random}, i.e., payoff $x$ is a multivariate random variable. The second game denoted by $(N,v)$ , i.e., game follows the model assumed in Theorem~\ref{thm:norm_dr}.
In this setup, when the correlation coefficient $\rho_{i,j}=1, \ \forall i,j\in N$, i.e., $x_i$ and $x_j$ are perfectly correlated for any players $i,j\in N$, and the variance of each $x_i$ is defined as $\textbf{Var}(x_i)\equiv\overline{\sigma}^2_i=r_i^2\sigma_N^2$, where $r_i\geq 0 \ \forall i\in N$ and $r(N)=1$, then $\textbf{DC}(z)=\DCdr$.  

The previous observation demonstrates that under the assumption of fully correlated payoffs, i.e.,\ $\rho_{i,j}=1, \ \forall i,j\in N$, the model employing the general allocation type can be described by the model with allocations with transfer payments. This assumption is intuitively aligned with the notion of coalition formation, where cooperating players are likely to exhibit correlated payoffs. Such a reasoning can be extended to allocation with transfer payment and general risk part $(d,r)$, where the pairwise correlations can be $-1$, which indicates perfectly inversely correlated payoffs.
This corollary supports our focus on studying coalitions where all players within the group exhibit strongly correlated values, with the pairwise correlations absolute value $1$. Such a setting helps in understanding the dynamics and payoff distributions within cooperating groups.

Further exploration into the stochastic payoffs in cooperative games might consider the implications of arbitrary correlations among the marginal distributions of the random payoffs. This approach, although potentially less tractable, raises interesting questions about the definition of coalition formation and the appropriate methods for distributing profits among various coalitions. The complexity of arbitrary correlations presents challenges in precisely defining what it means for a coalition to form and how its profits should be allocated.

\section{Undominated SSD-core}\label{app:undominated-ssd}
This section deals with the \emph{undominated SSD-core}, a generalization of the SSD-core introduced in the paper. We discuss the advantages and disadvatanges on several examples.
\begin{defn}[Undominated SSD-core]
    Let $(N,v)$ be a stochastic TU-game. The Undominated SSD-core is a set of efficient stochastic payoffs $x$ denoted by $\textbf{UDC}(v)$ for which it holds that
    $$ \not\exists S\subseteq N: x(S)\prec_{SSD}v(S) \And x(N)\text{ has the same distribution as }v(N).$$
\end{defn}
Similarly to the SSD-core we can distinguish among multiple allocations types. The undominated SSD-core concides with the SSD-core on complete orders, however covers more stochastic payoffs once the order is not complete. The following example  demonstrates that despite this advantage, the concept is far from perfect.
\begin{example}\label{example:undominated_core}
    Let $(N,v)$ be a 2-player stochastic TU-game with characteristic functions $v(\{1\})\sim N(10,1)$, $v(\{2\})\sim N(10,1)$, and $v(\{1,2\})\sim N(2,10)$.
    It is immediate that $\DCdr = \emptyset$ since there are no  $d_1$, $d_2 \in \R$ satisfying $$d_1+d_2=2 \And d_1\geq 10 \And d_2\geq 10.$$
However, the following set of payoffs lies in $\UDCdr$:
    $$(d_1,d_2,r_1,r_2)=(\alpha,2-\alpha,\beta,1-\beta),\ \text{where} \ \alpha>10,\beta>\frac{9}{10}$$
    and symmetrically for those stochastic payoffs where the roles of $1$ and $2$ is switched.
    We immediately see that $\UDCdr$ does not have to be bounded, and what is more, for large enough $\alpha$ and $\beta$, includes stochastic payoffs under which one player receives a much higher expected and the variance value than the other even though their role in the game is symmetric.
\end{example}
We note that despite the unboundedness in some of the scenarios, the undominated core is not guaranteed to be nonempty. To see this, consider a modification of Example~\ref{example:undominated_core}, where $v(\{i\}) \sim N(10,20)$ for $i=1,2$. Both of these examples argue against the use of this concept. In the next example, we outline a situation in which the undominated core might be reasonable.
\begin{example}
    Let $(N,v)$ be a stochastic TU-game with $v(S)\sim N(\mu_S,\sigma_S^2)$. In this game, the agents incur fixed costs $c_i$, and it is agreed that each agent will receive a payment of $d_i = \frac{c_i}{\sum_{k \in N}c_k}$. Now, before the value of the game is realized, they have to agree upon the rearrangement of the final payoff. A stochastic payoff from SSD-core might be desirable, however, vector $d$ might violate the conditions $d(S)\geq \mu_S$. However, any $r$ satisfying $\sigma_S>r(S)\sigma_N$, $\forall S\subseteq N$
    implies that together with $d$, vector $(d,r)$ lies in the undominated SSD-core.
\end{example}





\end{document}